\documentclass[11pt]{article}

\usepackage{fullpage,wrapfig}
\usepackage{amsmath, amsthm, amssymb, algorithm,thmtools,algpseudocode,enumerate,caption}

\usepackage[usenames,dvipsnames,svgnames,table]{xcolor}
\definecolor{darkgreen}{rgb}{0.0,0,0.9}

\usepackage[colorlinks=true,pdfpagemode=UseNone,citecolor=OliveGreen,linkcolor=BrickRed,urlcolor=BrickRed,
pagebackref]{hyperref} 

\usepackage{tikz}
\usetikzlibrary{arrows,math,calc}

\newcommand{\Input}{\item[{\bf Input:}]}
\newcommand{\Output}{\item[{\bf Output:}]}
\renewcommand{\Return}{\item[{\bf return}]}

\makeatletter
\newcommand{\setword}[2]{%
  \phantomsection
  #1\def\@currentlabel{\unexpanded{#1}}\label{#2}%
}
\makeatother

\def\bone{{\bf 1}}

\def\bx{{\bf x}}
\def\by{{\bf y}}
\def\bz{{\bf z}}
\def\bof{{\bf f}}
\def\R{\mathbb{R}}

\def\mE{\mathbb{E}}

\def\eps{\epsilon}

\newcommand{\EE}[2]{\mE_{#1}\left[#2\right]}

\DeclareMathOperator{\st}{s.t.} 
\DeclareMathOperator{\poly}{poly}

\def\h{h}

\def\expansion{expansion}
\newcommand{\Set}[2]{
  \{\, #1 \mid #2 \, \}
}
\newtheorem{theorem}{Theorem}[section]

\newtheorem{lemma}{Lemma}[section]
\newtheorem{fact}[theorem]{Fact}
\newtheorem{proposition}[theorem]{Proposition}

\newenvironment{proofof}[1]{{\em Proof of #1.}}{\hfill
\qed}
\DeclareMathOperator{\polylog}{polylog}

\title{Approximation Algorithms for Finding\\ Maximum  Induced Expanders}

\author{
Shayan Oveis Gharan
\thanks{Department of Computer Science and Engineering, University of Washington.
Email: \protect\url{shayan@cs.washington.edu}}
\and
Alireza Rezaei
\thanks{Department of Computer Science and Engineering, University of Washington.
Email: \protect\url{arezaei@cs.washington.edu}}
}
\begin{document}
\maketitle

\begin{abstract}
	We initiate the study of approximating the largest induced expander in a given graph $G$. Given a $\Delta$-regular graph $G$ with $n$ vertices, the goal is to find the set with the largest induced expansion of size at least $\delta \cdot n$. We design a bi-criteria approximation algorithm for this problem; if the optimum has induced spectral expansion $\lambda$ our algorithm returns a $\frac{\lambda}{\log^2\delta \exp(\Delta/\lambda)}$-(spectral) expander of size at least $\delta n$ (up to constants).

	Our proof introduces and employs a novel semidefinite programming relaxation for the largest induced expander problem. We expect to see further applications of our SDP relaxation in  graph partitioning problems. In particular, because of the close connection to the small set expansion problem, one may be able to obtain new insights into the unique games problem.
\end{abstract}



\section{Introduction}
In an instance of the maximum clique problem, we are given an undirected graph $G=(V,E)$ and the goal is to find the largest set $S$ of vertices of $G$ such that the induced subgraph $G[S]$ is a clique. The maximum clique problem is extensively studied in the last several decades and it is shown to be one of the hardest problems to approximate in the worst case \cite{Has96}.  

Although the maximum clique problem has many applications in theory and practice,  $G[S]$ being an actual complete graph is a property that is  unstable with respect to slight changes in $G$. First of all, there is no natural extension of the maximum clique problem to weighted graphs. Even if $G$ is unweighted, a large clique of $G$ may be completely eliminated by removing only a few edges of $G$.  For a concrete example, suppose $G$ is a complete graph, i.e., the maximum clique of $G$ has size $n$; if we delete only an $o(1)$ fraction of edges of $G$ uniformly at random, the size of the maximum clique of $G$ reduces by an exponential factor to $\polylog(n)$ \cite{GM75,BE76}. 

It is a natural question to find the maximum size subgraph of $G$ such that $G[S]$ is ``clique-like''. There are several directions to formalize the clique-like property of $G[S]$: For example, one can say $G[S]$ is clique-like if the local neighborhood of every vertex is similar to a clique, i.e., if the average degree of vertices in $G[S]$ is $\Omega(|S|)$; such a measure corresponds to the densest subgraph problem which is also extensively studied in the past decade \cite{Cha00,Fei02,Kho04,BCCFV10}.

In this paper, we use \emph{spectral expansion}
as a \emph{global} clique-like property. First, we define spectral expansion, and then we justify that it can be considered as a clique-like property.
Let $G$ be a $\Delta$-regular graph with $n=|V|$ vertices.
For a pair of vertices $u,v\in V$, let $\bone_{u,v}\in \R^{V}$ be the vector that is $1$ in $v$, $-1$ in $u$ and zero everywhere else. Let $L_{u,v}=\bone_{u,v}\bone_{u,v}^\intercal$. The Laplacian of $G$, $L_G$ is defined as follows:
$$ L_G=\sum_{u\sim v} L_{u,v},$$
where we write $u\sim v$ to denote $\{u,v\}\in E$. 
Note that if $G$ is weighted, then we need to scale $L_{u,v}$ with the weight of the edge $\{u,v\}$.
It is easy to see that $L_G$ is a PSD matrix, its first eigenvalue is zero, and the corresponding eigenvector is the all-ones vector.
The spectral expansion of $G$ is defined as the second smallest eigenvalue  of $L_G$, $\lambda_2(L_G).$
We say $G$ is an $\eps$-\emph{expander} if $\lambda_2(L_G)\geq \eps$.
In this paper,  we design bicriteria approximation algorithms for approximating the largest induced  expander of $G$. 

It is a well-known fact that ($\Delta$-regular)  $\Omega(\Delta)$-expander graphs are essentially sparse complete graphs.
This can be justified    by analyzing either the spectral or combinatorial properties of expander graphs. The eigenvalues of the Laplacian matrix of an $\Omega(\Delta)$-expander graph are essentially the same as the eigenvalues of a complete graph scaled by $\Delta/n$. 
Similarly, the size of any cut is (up to constants) equal to $\Delta/n$ fraction of the same cut in the complete graph. This  follows from the Cheeger's inequality.
For a set $S\subseteq V$ let
$$ h(S)=\frac{|E(S,\overline{S})|}{ |S|},$$
be the \emph{combinatorial expansion} of $S$, where $E(S,\overline{S})=\{\{u,v\}: u\in S, v\notin S\}$ is the set of edges in the cut $(S,\overline{S})$. The combinatorial \expansion~of $G$, $\h(G)$, is defined as follows:
$$ \h(G) = \min_{\emptyset \subsetneq S\subsetneq V} \max\{\h(S),\h(\overline{S})\}=\min_{\emptyset \subsetneq S\subsetneq V} \frac{|E(S,\overline{S})|}{\max\{|S|,|\overline{S}|\}}.$$
 Cheeger's inequality relates (combinatorial) expansion to the spectral expansion.
\begin{theorem}[Discrete Cheeger's inequality \cite{AM85,Alon86}]
\label{thm:cheeger}
For any graph $G$ with maximum degree $\Delta$, we have
	\begin{equation}\label{eq:cheeger} \frac{\lambda_2(L_G)}{2} \leq h(G) \leq \sqrt{2\Delta \lambda_2(L_G)}.	
	\end{equation}
\end{theorem}
By the above theorem, if $G$ is an $\Omega(\Delta)$-expander, then $\h(G)\geq \Omega(\Delta)$. 
One can also prove tighter connections between the structure of cuts in an expander graph and the complete graph by  the expander mixing lemma and its generalizations (see e.g., \cite{BL04}).
In summary, unlike the density, (spectral) expansion can be considered as a global clique-like property.
\paragraph{Motivations.} 
Variants of the largest induced expander problem are previously studied and employed in the design of approximation algorithms. Trevisan \cite{Tre05} showed that one can remove a small fraction of the edges of  $G$ such that any connected component of the remaining graph is an $\Omega(\Delta/\polylog(n))$-expander. He used this fact to design an approximation algorithm for the unique games problem.
More recently, the first author together with Anari showed that if $G$ is $\Delta$-edge-connected, then it has an induced $\Omega(\Delta)$-edge-connected subgraph that is an $\Omega(\Delta/\log n)$-expander.
This fact is used to design an approximation algorithm for Asymmetric TSP \cite{AO14,AO14b}. We emphasize that both of the aforementioned results do not provide any guarantee on the size of the induced expanders that they construct.

Finding induced expanders can also have practical applications in clustering and community detection problems. 
Classically, the expansion or \emph{conductance} are used as combinatorial measures for the quality of a clustering of a graph.
This parameter fails dramatically when the underlying clusters are \emph{overlapping} because the (outside) expansion of each cluster is $\Omega(1)$. The failure of using sparsest cut approximation algorithms is one of the major challenges in overlapping clustering. In those scenarios, it is more natural to look for a cluster which induces an expander graph. For a concrete example, consider the set of all people living in USA in a world wide social network. Since each person typically belongs to multiple international communities, such a set has a large outside expansion. However, it is expectable that it  induces an $\Omega(1)$-expander. 

In general, unlike outside expansion, if $G[S]$ has $\Omega(\Delta)$ spectral expansion, then it has many properties which resembles the structure of a community: 
\begin{enumerate}[i)]
\item Low degree of separation: The diameter of $G[S]$ is at most $O(\log |S|)$. 
\item Small mixing time of random walks: A simple random walk in $G[S]$ mixes in time $O(\log |S|)$ (see \cite{LPW06} for the definition of mixing time and its connection to expander graphs).
\end{enumerate}
We refrain from going into the detailed properties of expanders and we refer interested readers to \cite{HLW06}. 
Next,  we formally define our problem and its SDP relaxation, then we describe our results.

\paragraph{Problem Formulation.}
Throughout the paper we assume that $G=(V,E)$ is an  undirected, unweighted, $\Delta$-regular graph.
We restrict our attention to unweighted graphs for the brevity of the arguments, but all of our results naturally extend to weighted graphs.
Given a parameter $\delta$, we are interested in finding a subset $S\subset V$ of size $|S|\geq \delta n$ with the largest induced spectral expansion,
\begin{equation}\label{eq:benchmarkinducedexp} \lambda(\delta) := \max_{S: |S|\geq \delta n} \lambda_2(L_{G[S]})
\end{equation}
The interesting regime of the problem is when $\lambda(\delta)=\Omega(\Delta)$, i.e., when $G$ has a sparse complete graph as a subgraph.
Because of this, our goal is to approximate the above objective function with no (or as little as possible) loss on the size of $G$ and $\delta$. Our approximation factor may have an exponential loss in $\lambda(\delta)$.

The above problem can be considered as a ``dual'' of the \emph{small set expansion} problem \cite{RS10}. In an instance of the small set expansion problem we are given a $\Delta$-regular graph, and a parameter $\delta$ and we want to find the set $S$ of size at most $\delta\cdot n$ with the smallest (outside) expansion, i.e., we want to approximate 
$$ \h(\delta):=\min_{S: |S|\leq \delta n} \h(S).$$
The problem is extensively studied in the last couple of years because of its close connection to the unique games problem \cite{RST10,BFKM11,OT12,KL12}. To this date, all of the  approximation algorithms of the small set expansion problem incur a loss $\poly(\log(1/\delta))$ in the expansion of the output. The following simple fact relates the two problems
\begin{fact}
	If $G$ has a partitioning into $2\delta\cdot n$ sets each inducing an $\Omega(\eps\cdot \Delta)$-expander,  then 
	$$\h(\delta)\geq \Omega(\eps\cdot\Delta).$$
\end{fact}
Because of the above close connection to the small set expansion problem, our SDP relaxation and the rounding algorithm also incur a $\poly(\log(1/\delta))$ loss in the expansion of the output.

\paragraph{Related Works.}
In the last decade three general families of algorithms are studied to detect communities which have large induced expansion. The first one is the class of  greedy based algorithms, the second one is the family of local random walk based algorithms, and the last one is the  spectral algorithms that employ eigenvectors of the Laplacian matrix. To the best of our knowledge, all of these algorithms  fail to capture an induced expander because the vertices of the expander may be highly connected to the outside, i.e., we may have $h(S)\geq 1/2$.
The failure is because of the fact that these algorithms are specifically designed to detect sets with small (outside) expansion.

Let us elaborate on the latter fact in each of the three cases. Greedy based algorithms \cite{KVV04,Tre05,AO14} recursively partition the graph using an approximation algorithm for the sparsest cut problem; the algorithm stops once there is no sparse cut in any set of the partition. If the vertices of the hidden expander are highly connected to the outside, the algorithm may simply separate them apart and the structure of the expander will be lost in the partitioning of the graph. Nonetheless, we show that a variant of this algorithm provides an $\poly(\delta)$-approximation to the largest induced expander problem; we will also provide some tight examples.

Local graph clustering algorithms \cite{ST08,ACL06,AP09,OT12,ZLM13} simulate simple lazy random walks, or the associated Markov chains like the page rank \cite{ACL06} or the evolving set process \cite{MP03}, on a graph. They detect a nonexpanding set by looking at threshold sets of the probability distribution of the walk at some time $t$.
Perhaps, the closest result to our work is the  work of Zhu, Lattanzi and Mirrokni \cite{ZLM13} who show that if for a set $S$, $\lambda_2(G[S]) \gg \log(n)h(S)$,   then it is possible to recover the set. Unfortunately, when $h(S)$ is large, the random walk algorithm fails to recover $S$ because before the walk visits all vertices of $S$, most of the probability mass has escaped the set.

The last family of algorithms use spectral methods, in particular the eigenvalues and eigenvectors of the (normalized) Laplacian matrix to detect the communities \cite{LOT12,LRTV11,OT14,DPRS14,PSZ15,Sin16}. These algorithms typically assume that there is a large gap between the $k$ and $k+1$ eigenvalue of the graph. This assumption implies that the graph can be partitioned into induced expanders which have very small outside expansion \cite{OT14}. It follows that by utilizing the first $k$ eigenvectors of the Laplacian matrix one can recover these expander graphs. However, the existence of a large size induced expander (possibly with large outside expansion) does not  guarantee the existence of small eigenvalues, so, in our settings, the spectral methods fail to recover the hidden expander.


\subsection{Our SDP relaxation}
As alluded to in the previous section, the known local and spectral algorithms fail to find a large induced expander in a given graph. Therefore, in this work, we use semidefinite programming to write a convex relaxation of \eqref{eq:benchmarkinducedexp}. 
There are two underlying obstacles to write a SDP relaxation for our problem. Firstly, the local neighborhood of a vertex in an expander graph may be very sparse and look like just a tree. Therefore, unlike the Lov\'asz theta function \cite{Lov79}, $G$ being an expander does not enforce any constraints on the local neighborhoods. 
Secondly, having an induced expander of size say $\sqrt{n}$ does not imply any global constraint on the structure of $G$. So, our SDP constraints must be ``localized''   to the induced expander that we are trying to find.


Before describing the relaxation, we need to set up a notation and write an equivalent definition of expander graphs.
For a symmetric matrix $A\in \R^{V\times V}$, we say $A$ is positive semidefinite (PSD), $A\succeq 0$, if for any vector $\bx\in\R^V$,
$$ \bx^\intercal A \bx\geq 0.$$
For two matrices $A,B\in\R^{V\times V}$, we write $A\succeq B$ if $A-B$ is PSD. 

Fix a set $S\subseteq V$ and let $K_S$ be a complete graph induced on $S$. 
It is a simple fact that all (except the first) eigenvalues of the Laplacian matrix of a complete graph of size $n$ are equal to $n$. 
Since all (except the first) eigenvalues of $L_{G[S]}$ are at least $\lambda_2(L_{G[S]})$ we can write
\begin{equation}\label{eq:GKSlambda2} L_{G[S]}\succeq \frac{\lambda_2(L_{G[S]})}{|S|}\cdot L_{K_S},
\end{equation}
Next,  we use the above simple inequality to write our SDP relaxation of \eqref{eq:benchmarkinducedexp}.
See \ref{sdp:noh} for the details of our SDP relaxation.
\begin{figure*}
\centering
\fbox{\parbox{5in}{ \vspace*{0mm}
\begin{center}{\bf \setword{SDP 1}{sdp:noh}}\end{center}\vspace{-.7cm}
\begin{align}
	\max \hspace{3ex} & \hspace{1ex}\lambda, & \nonumber \\
	\st \hspace{4ex}
	& \sum_{e\in E} x_{e}L_{e} \succeq \lambda \cdot \sum_{u,v \in V} y_{\{u,v\}}L_{u,v}, &\label{eq:lambdaconstraint} \\
	& \sum_{v} y_{\{u,v\}} \geq \max_{e\sim u} x_e \,  & \forall  u \in V, \label{eq:sumyconstraint} \\
	& \delta n \cdot y_{\{u,v\}}  \leq \sum_{w} y_{\{u,w\}}& \forall   u,v\in V,  \label{eq:maxyconstraint} \\
	& \bx,\by \geq 0, \bx\neq 0. & \nonumber
\end{align}
}}
\end{figure*}
Note that the first constraint of the SDP, \eqref{eq:lambdaconstraint}, is not convex. To make it convex, it is enough to solve the SDP with an explicit value of $\lambda$, and then run a binary search to maximize $\lambda$.

Let us show that \ref{sdp:noh} is a relaxation of \eqref{eq:benchmarkinducedexp}, i.e., its optimum value is at least $\lambda(\delta)$.
Let $S\subseteq V$ such that $|S|\geq \delta\cdot n$ be the set maximizing \eqref{eq:benchmarkinducedexp}.

Our intended integral solution is defined as follows:
 We let $x_e=1$ if both endpoints of $e$ are in $S$ and zero otherwise, and we let $y_{\{u,v\}}=\frac1{|S|}$ if $u,v\in S$ and zero otherwise. 
Let us verify the first constraint of the SDP and the rest are easy to check. It is easy to see that 
$\sum_e x_e L_e = L_{G[S]}$ is the Laplacian of the induced graph $G[S]$. On the other hand, $\sum_{u,v} y_{\{u,v\}}L_{u,v}=\frac{L_{K_S}}{|S|}$
is the Laplacian of a complete graph  on $S$ scaled by $\frac1{|S|}$. 
Therefore, the first constraint of the SDP follows by \eqref{eq:GKSlambda2}.

We can strengthen the above relaxation (and our results) when the optimum induced expander is loosely connected to the outside. That is, suppose the optimum set $S$ of \eqref{eq:benchmarkinducedexp} satisfies
$\h(S) \leq \h^*.$
Let $E(S)=\{\{u,v\}:u,v\in S\}$ be the set of edges between the vertices of $S$.
Then, by the above inequality,
$$ |E(S)|=\frac{\Delta\cdot |S|-|E(S,\overline{S})|}{2} \geq |S|\cdot \frac{(\Delta-\h^*)}{2}.$$
So, we can strengthen \ref{sdp:noh} by adding a relaxation of the above inequality. See \ref{sdp:withh} for the new SDP. 
Note that, although the constraint \eqref{eq:hconstraint} is nonlinear, we can make it linear by introducing new variables $\{z_u\}_{u\in V}$ where  $x_e\leq z_u$ for all $e\sim u$.
It is an easy exercise that for a set $S\subset V$, the vector solution $\bx,\by$ that we constructed in the preceding paragraphs satisfy constraint \eqref{eq:hconstraint}.
\begin{figure*}
\centering
\fbox{\parbox{4in}{ 
\begin{center}{\bf \setword{SDP 2}{sdp:withh}}\end{center}\vspace{-.7cm}
\begin{align}
	\max\hspace{3ex}& \hspace{1ex}\lambda, &\nonumber\\
	\st\hspace{4ex}& \bx,\by \text{ satisfy constraints of \ref{sdp:noh},} & \nonumber \\
	& \sum_{e} x_e \geq \frac{(\Delta-\h^*)}{2}\cdot\sum_u \max_{e\sim u} x_e. & \label{eq:hconstraint}
\end{align}
}}
\end{figure*}

\subsection{Our Results}
\label{subsec:ourresults}
In this subsection we describe the main results of this paper. 
Before describing our main result, we design a simple greedy algorithm analogous to the work of Kannan, Vempala and Vetta \cite{KVV04} (and \cite{Tre05,AO14}) for the largest induced expander problem. 
\begin{theorem}
\label{thm:greedy}
There is a polynomial time algorithm that for any $\Delta$-regular graph $G$, $\delta < 1$, returns a set $S$ of size $|S| \geq 3\delta\cdot n/8$ and spectral expansion 
$$\lambda_2(L_{G[S]}) \gtrsim \frac{\delta^2\cdot \lambda^2}{\Delta\log^2 \delta},$$
where $\lambda = \lambda(\delta)$.
\end{theorem}
The algorithm simply uses repeated applications of the spectral minimum bisection algorithm to find an induced expander. See \autoref{sec:greedy} for the proof of the above theorem. The main downside of the above result is the polynomial dependency on $\delta$ which is essential to the greedy algorithm (see \autoref{prop:tightgreedy}). 
In particular, if $\delta=1/\sqrt{n}$ (and $\Delta=O(1)$), any connected subgraph of $G$ of size $\Omega(\sqrt{n})$ is a $\delta^2$-expander.

In our main result, we use \ref{sdp:noh} to exponentially improve the polynomial dependency on $\delta$ in the greedy algorithm. 
We design a bicriteria approximation algorithm for $\lambda(\delta)$; we show that any feasible solution of \ref{sdp:noh} can be rounded to a set of size $\Omega(\delta n)$ and spectral expansion $\frac{\lambda(\delta)}{\log^2\delta \cdot \exp(\Delta/\lambda)}$.

\begin{restatable}{theorem}{mainthm}
\label{thm:maintheoremwithoutoutsideexpansion}
There is a polynomial time algorithm that for any $\Delta$-regular graph $G$, $\delta>0$ and any feasible solution $\lambda,\bx,\by$ of \ref{sdp:noh},
 returns a set $S$ of size $|S|\geq 3\delta\cdot n/8$ and spectral expansion
$$ \lambda_2(L_{G[S]}) \gtrsim \dfrac{{\lambda}}{\log^2 \delta\cdot \exp(O(\Delta/\lambda))},$$ 
\end{restatable}
In the regime where $\lambda(\delta)=\Omega(\Delta)$ the approximation factor of the above theorem is $\log^2\delta$ (up to constants).
As a simple corollary, because of logarithmic dependency on $\delta$, we can use the above algorithm to find an $\Omega(\Delta/\polylog(n))$-expander of size $n^{\Omega(1)}$ in $G$ assuming the existence of an $\Omega(\Delta)$-expander of a similar size.

The $\log(1/\delta)$ loss in the above theorem essentially follows because of the connection to the small set expansion problem.
To make this connection more rigorous, we complement the above theorem and we show that, assuming $\Delta$ is sufficiently large, 
 the integrality gap of \ref{sdp:noh} is at least $\Omega(\log 1/\delta)$.
\begin{theorem}
\label{thm:IG1}
The integrality gap of the \ref{sdp:noh} is $\Omega (\min(\log 1/\delta,\Delta))$.
\end{theorem}

Our integrality gap example is made up of a hypercube of $\log(1/\delta)$ dimensions where every vertex is blown up to a cloud of $\delta n$ vertices. For every edge of the original hypercube, we add a complete bipartite graph of weight $1/\delta n$ between the vertices of the corresponding clouds.

	

Furthermore, we show that in certain regimes we can improve the exponential dependency on $\Delta/\lambda$ assuming the optimum solution of the largest induced expander problem has a small (outside) expansion. 
\begin{restatable}{theorem}{thmmainnoexp}
\label{thm:mainwithoutsideexpansion}
	There is a polynomial time algorithm that for any $\Delta$-regular graph, $\delta>0$, $h^*\leq \Delta(1-\frac2e)$, and any feasible solution $\lambda,\bx,\by$ of \ref{sdp:withh}, returns a set $S$ of size at least $|S|\geq 3\delta \cdot n/8$, and spectral expansion
	$$ \lambda_2(L_{G[S]}) \gtrsim \frac{\lambda^2}{\Delta\log^2 \delta\cdot (1+2h^*/\Delta)^{O(\Delta/\lambda)}}. $$
\end{restatable}
As a corollary of the above theorem, assume that $h^*\leq O(\eps\cdot\Delta)$ and $\lambda(\delta)\geq \Omega(\eps\cdot \Delta)$, i.e., there is a set $S$ of size $|S|\geq \delta n$ such that $\lambda_2(L_{G[S]})=\lambda(\delta)\geq \Omega(\eps\cdot\Delta)$ and $h(S)\leq  O(\eps\cdot \Delta)$. Then, by the above theorem, in polynomial time we can find a set $T$ of size $\Omega(\delta n)$ such that
$$ \lambda_2(L_{G[T]}) \gtrsim \frac{\eps\cdot \lambda(\delta)}{\log^2(\delta)}.$$


\subsection{Preliminaries}
Throughout the paper, we use bold letters to represent vectors. 
Unless otherwise specified, we let $\bx,\by,\lambda>0$ 
represent a feasible solution of \ref{sdp:noh}. Note that since feasible solutions of \ref{sdp:withh} is a subset of feasible solutions of \ref{sdp:noh}, any result for feasible solutions of \ref{sdp:noh} extends to the solutions of \ref{sdp:withh}.
Without loss of generality, we extend $\bx$ to all unordered pairs $\{u,v\}$, and we let  $x_{\{u,v\}}=0$ whenever $\{u,v\}\notin E$.

For two disjoint subsets of vertices $S,T \subseteq V
$, we let 
$$E(S,T):=\{\{u,v\} \in E: u \in S, v \in T\}$$
be the edges connecting $S$ to $T$.
For a vector $\bx \in \R^{V\times V}$, we let  
$\bx(S,T):=\sum_{u\in S, v\in T} x_{\{u,v\}}$. 
We use $G_\bx$ to denote the  graph with vertex set $V$ where the weight of the edge connecting each pair of vertices $u,v$ is $x_{\{u,v\}}$.
Similarly, we use $G_\by$ to denote the graph weighted by vector $\by$. 

For any vertex $v \in V$, let  
$$z_v:=\max_{e \sim u} x_e$$
be the weight of $v$.
Observe that if $z_v=0$ then all edges incident to $v$ have weight $0$. It is easy to see that any feasible solution of the SDP remains feasible when we delete all vertices of weight zero.
Therefore, throughout the paper we assume that $z_v>0$ for all $v\in V$.

We define the \emph{width} of $S \subseteq V$ to be $\frac{\max_{v \in S} z_v}{\min_{v \in S}z_v}$. 
The weighted expansion of a set $S\subseteq V$ in $G_\bx$ (and $G_\by$) is the ratio of the sum of the weights of the edges in the cut $(S,\overline{S})$ to the sum of the 
weights of vertices of $S$,
$$ h_\bx(S) = \frac{\bx(S,\overline{S})}{\bz(S)}, h_\by(S)=\frac{\by(S,\overline{S})}{\bz(S)}.$$

\subsection{Background on spectral graph theory}
\label{sec:spectralgraphtheory}
Perhaps the most natural property 	of the Laplacian matrix is the simple description of their quadratic form. For any vector $\bof \in \R^V$, 
\begin{equation*}
\bof^\intercal L_G\bof = \sum_{\{u,v\} \in E} (f_u-f_v)^2.
\end{equation*}
Note that if $G$ is weighted every term in the RHS will be scaled by the weight of the edge $\{u,v\}$.
One simple consequence of the above identity is that the Laplacian is always a PSD matrix. 
A simple application of the above identity is that we can write the size of a cut $|E(S,\overline{S})|$ as a quadratic form. For $\bof=\bone_S$ we get,  
$$
|E(S,\overline{S})|=\bone_S^\intercal L_G \bone_S. 
$$ 

As alluded to in the introduction, the Cheeger's inequality relates the second eigenvalue of the Laplacian matrix to $\h(G)$. The left side of \eqref{eq:cheeger} is known as the \emph{easy} direction, and the the right side is the \emph{hard} direction. The proof of the hard direction follows by a simple rounding algorithm known as the \emph{spectral partitioning algorithm} which rounds the second eigenvector of the Laplacian matrix to a set $S$ of (size $|S|\leq |V|/2)$ and) expansion $O(\sqrt{\lambda_2\cdot \Delta})$. 
For the sake of completeness, here we describe the algorithm:
Let $\bof$ be the second eigenvector of $L_G$. Sort vertices based on $f_v$, and call them $v_1,v_2,\dots,v_n$. Return the best threshold cut, i.e., 
$$ \min_{1\leq i\leq n} \max(h(\{v_1,\dots,v_i\}), h(\{v_{i+1},\dots,v_n\})).$$

One can use repeated applications of the preceding algorithm to approximate the minimum bisection of a given graph $G$. See \autoref{alg:minbisection} for the details of the algorithm. 
\begin{lemma}
\label{lem:bisection}
 Let $G=(V,E)$ be a graph with maximum degree $\Delta$. For every $0< \epsilon <1$, Algorithm \ref{alg:minbisection} returns a set $S \subseteq{V}$ such that either $|S| \geq \frac{3}{4} |V|$ and $\lambda_2 (G[S]) \geq \epsilon \cdot \Delta$, or $\frac{|V|}{4} \leq |S| \leq \frac{3|V|}{4}$ and $\h(S) \leq \sqrt{2\epsilon}\cdot\Delta$.
\end{lemma}
The proof of the above lemma simply follows from \autoref{thm:cheeger} and the fact that for any two disjoint sets $S,T$, $h(S\cup T)\leq \max(\h(S),h(T))$.
\begin{algorithm}
\begin{algorithmic}[1]
	\Input A graph $G=(V,E)$ with maximum degree $\Delta$ and $0<\epsilon <1$.
	\Output A set $S$ s.t., either  $|S| \geq \frac{3|V|}{4}$ and $\lambda_2(G[S])\geq \epsilon\Delta$, or  $\frac{|V|}{4} \leq |S| \leq \frac{3|V|}{4}$ and $h(S) \leq \sqrt{2\epsilon}\cdot \Delta.$
	\State Let $S\leftarrow V$.
	\While{ $|S| \geq \frac{3|V|}{4}$}	
		\State If $\lambda_2(G[S]) \geq \epsilon\cdot \Delta$ then \textbf{return} $S$.
			 \label{line:returnexpander} 
		\State Otherwise, run the spectral partitioning on $G[S]$ and let $(T,S\setminus{T})$ be the output.
		 	\State Say $|T|<|S\setminus T|$. 
		 	let $S \leftarrow S \setminus{T}$. 
	\EndWhile
	\Return $V\setminus S$.
\end{algorithmic}
\caption{Spectral Bisection Algorithm}
\label{alg:minbisection}
\end{algorithm}


\section{Proof Overview}\label{sec:proofoverview}
Let $G_\bx,G_\by$ be the graphs weighted by the $\bx$ and $\by$ vectors respectively. In the first step of the proof, we exploit the main constraint of the SDP, i.e., \eqref{eq:lambdaconstraint}, to show that $G_\bx$ is a $\lambda/2$-small set  weighted expander, i.e., every set $S$ of size  $|S|\leq \delta n/2$ satisfies $\h_\bx(S) \geq \lambda/2$.
Although the proof of this statement is simple, it crucially uses the SDP constraints.
Firstly, we use \eqref{eq:lambdaconstraint} to show that for any set $S$, $ \h_\bx(S) \geq \lambda\cdot  \h_\by(S).$ Then, we use constraints \eqref{eq:sumyconstraint} and \eqref{eq:maxyconstraint} to show that $G_\by$ is a $1/2$-small set  weighted expander; this implies that $G_\bx$ is a $\lambda/2$-small set weighted expander (see \autoref{lem:expandingsmallsets} for the details of the proof). 
This statement  enlightens  a deep connection between our SDP and the small set expansion problem which may have further applications in understanding the computational complexity of the small set expansion problem. 

In the second step, we  essentially reduce the problem to the case where $\bz$ is almost a constant vector. 
The consequence is that when $\bz$ is a constant vector,  the weighted expansion is the same as (unweighted) expansion up to a normalization. Therefore,  we can conclude from the previous paragraph that $G$ is a small set expander.
More precisely, in the second step, we find a set $S\subseteq V$ of small width such that $\h_\bx(S)\ll \lambda$. Note that any such set must satisfy $|S|\geq \delta\cdot n/2$. Since $S$ has a small width, the $\bz$ vector restricted to the induced graph $G[S]$ looks like a constant vector. 
If $\phi_\bx(S)=0$, then indeed $G_\bx[S]$ is a small set expander. 
But, if $\phi_\bx(S)\neq 0$, 
 we cannot conclude that \emph{any} small set $T\subseteq S$ has a large unweighted expansion. 
Nonetheless, since $\h_\bx(S)$ is small, a random small set has large unweighted expansion; in particular, if we partition $S$ into many small sets say $\{T_1,T_2,\dots\}$, we can conclude that 
\begin{equation}\label{eq:hGxbigTi} \EE{i}{\h_{G[S]}(T_i)}\approx \EE{i}{\h_{G_\bx[S]}(T_i)} \gtrsim \lambda.	
\end{equation}
This fact will be crucially used in the third step to find an induced expander.

To find  $S$ we run the following randomized algorithm: First we map each vertex $v$, to the point $\log z_v$ on the real line. Then, we randomly choose vertices in a window of length $w$, where the probability of each window is proportional to the total weight of the vertices that it contains.  By construction, the width of any set in the distribution is at most $e^w$; we use an averaging argument to show that the expected weighted expansion of a random window is proportional to $1/w$ (see \autoref{lem:nonexpandingset} for the details of the proof)

In the last step of the proof we design an algorithm to find an expander $G[T]$ in the set $S$ that we found in the previous step. 
We use the spectral bisection algorithm to recursively partition $G[S]$ until we find an $\eps\cdot \Delta$-expander, or  the size of every set in the partition is less than $\delta n/2$. 
It follows that a random set in the final partition has  unweighted expansion $O(\sqrt{2\eps}\cdot  \Delta \log(1/\delta))$. Since $S$ has width $e^w$, the \emph{weighted} expansion of any subset $T\subset S$ is within $e^w$ of its unweighted expansion. 
But, by  \eqref{eq:hGxbigTi} 
a random set in the final partition must have a weighted expansion at least $\Omega(\lambda)$. 
Letting $\eps\asymp \frac{1}{\log^2\delta \exp(\Delta/\lambda)}$ proves the theorem.
\section{The Analysis of the Simple Greedy Algorithm} 
\label{sec:greedy}
In this section we prove \autoref{thm:greedy}.
First we prove the following simple lemma.

\begin{lemma}
\label{lem:graphdecomposition}
There is a 
polynomial algorithm (Algorithm \ref{alg:findingexpander}) that for every graph $H$ with $n$ vertices of maximum degree $\Delta$ and parameters $0 < \epsilon,\delta < 1$, returns one of the followings. 
\begin{enumerate}[i)]
\item \label{case:largeexpander} A set $S \subseteq V(H)$ of size at least 
$3\delta \cdot n/8$ and $\lambda_2(H[S]) \geq \eps\cdot \Delta$ 
\item \label{case:decomposition} A
partition $\mathcal{P}$ of $V(H)$ into sets of size at most $\delta n /2$ such that  
\begin{equation}
\label{eq:numberofedges}
\sum_{T \in \mathcal{P}} |E(T,V(H)\setminus{T})| \leq 2\Delta(\log \frac{1}{\delta})\sqrt{2\epsilon }\cdot n.
\end{equation}
\end{enumerate}

\end{lemma}
\begin{proof}
If  $\lambda_2(H) 
\geq \epsilon\cdot  \Delta$ then we are done. Otherwise, we split 
$V(H)$ into two pieces by the spectral bisection algorithm introduced in \autoref{lem:bisection} for $\epsilon$. Then, we recursively run the bisection algorithm  
on each new set until we find either an $\epsilon \Delta$-expander, or all sets have size at most $\frac{\delta n}{2}$. The details are described in  \autoref{alg:findingexpander}. If we find an $\epsilon\Delta$-expander (Line \ref{line:elseline} of  \autoref{alg:findingexpander}),  its size is at least $ \frac{3}{4}\cdot 
 \frac{\delta n}{2}=\frac{3\delta n}{8}$, and we are done. 

\begin{algorithm}
\begin{algorithmic}[1]
	\Input A graph $H$ with maximum degree $\Delta$ and parameters $0 < \delta,\epsilon < 1$.
	\Output A subset of $V(H)$ or a partitioning of it. 
	\State Let $\mathcal{P}=\{V(H)\}$. 
	\While{ there is a set in $\mathcal{P}$ with more than  $\frac{\delta n}{2}$ vertices}
	    \ForAll { $S \in \mathcal{P}$ with $|S| > \frac{\delta n}{2}$}
	        \State Run \autoref{alg:minbisection} on input $\epsilon$ and $H[S]$. Let $T \subseteq S$ be the output.
	        \State If $\lambda_2(H[T]) \leq \epsilon\cdot \Delta$, return $T$.
	        Otherwise, add $T$ and $S\setminus{T}$ to $\mathcal{P}$ and remove $S$. \label{line:elseline}
		    \EndFor
	\EndWhile	
		\State Return $\mathcal{P}$.\label{line:nonexpanding}
\end{algorithmic}
\caption{Algorithm for finding either a large expander or a sparse partition}
\label{alg:findingexpander}
\end{algorithm}
Otherwise, Let $\mathcal{P}$ be the partition of $V(H)$ at the end of the algorithm. 
In this case, by description of the  algorithm all sets in $\mathcal{P}$ have size at most $\frac{\delta n}{ 2}$, so all we need to do is to prove \eqref{eq:numberofedges}.
Let $\mathcal{P}_i$ be the set $\mathcal{P}$ at the end of iteration $i$ of the main 
loop of the algorithm and define $e_i := \sum_{T \in \mathcal{P}_i} |E(T,V(H)\setminus{T})|$.    
By  description of the algorithm, we have the following two simple facts.
\begin{fact}
For any $i>1$, $e_{i} \leq e_{i-1}+\Delta \sqrt{2\epsilon}\cdot V(H)$. 
\end{fact}
The above holds since $\mathcal{P}_i$ is obtained by splitting all sets in $\mathcal{P}_{i-1}$ into two new sets by a cut of expansion at most $\sqrt{2\epsilon}\cdot  \Delta$.
\begin{fact}
The number of iterations of the main loop is at most $2\log\frac{1}{\delta}$.
\end{fact}
To see this, note that the algorithm terminates after $i$ steps where $i$ is the smallest number for which all the sets in $\mathcal{P}_i$ have size at most $\frac{\delta n}{2}$. Furthermore, in every iteration we split every set into two pieces, each of them having at most $\frac{3}{4}$ fraction of the vertices of the initial set. 
Combining these two facts, we get \eqref{eq:numberofedges} which completes the proof. 
\end{proof}

\begin{proofof}{\autoref{thm:greedy}}
We show that if for some $\eps$ the output of \autoref{alg:findingexpander} for $G,\delta,\eps$ is Case \ref{case:decomposition}, then 
\begin{equation}
\label{eq:contradictionwithnumofedges}
\eps \geq \frac{\lambda^2\delta^2}{32\Delta^2\log^2\delta}.
\end{equation}
So, to find an induced expander, it is enough to run \autoref{alg:findingexpander} for an $\eps$ smaller than the RHS.
Suppose that for some $\eps>0$ the algorithm returns a partition $\mathcal{P}$ of 
$V(G)$ satisfying Case \ref{case:decomposition}. 
By definition 
of $\lambda(\delta)$, there is a set $S\subseteq V$ of size $|S| \geq \delta n $ 
such that $\lambda_2(G[S]) \geq \lambda(\delta)$. So we have  
\begin{eqnarray*}
\label{eq:numberofedgescontradict}
\sum_{T \in \mathcal{P}} |E(T,V(G)\setminus{T})| &\geq& \sum_{T \in \mathcal{P}} |E( T \cap S,V(G)\setminus{T})| \\
&\geq& \sum_{T \in \mathcal{P}} \frac{\lambda}{2}|S \cap T|  \geq \frac{\lambda}{2}\cdot \delta n 
\end{eqnarray*}
where in the second inequality we use Cheeger's inequality and the fact that for every $T \in \mathcal{P}$, $|T\cap S| \leq |S|/2$ as $|T| \leq \delta \cdot n /2$ by Case \ref{case:decomposition} of the 
lemma.
Using \eqref{eq:numberofedges}, 
we get that
$$ \frac{\lambda}2\cdot \delta \leq 2\Delta\log(1/\delta)\sqrt{2\eps}\cdot n,$$
which proves \eqref{eq:contradictionwithnumofedges}.
\end{proofof}
\newline

In the following proposition we show that our analysis in the preceding theorem is essentially 
tight and the largest induced expansion that Algorithm \ref{alg:findingexpander} guarantees is $O(\delta \lambda(\delta))$. 
\begin{proposition}\label{prop:tightgreedy}
For any $0 < \delta < 1$, there exists a graph $G$ which is $O(1)$-regular such that the output of the algorithm of \autoref{thm:greedy} on input $G$ and $\delta $ is an  $O(\delta \lambda(\delta))$-expander. 
\end{proposition}
\begin{proof}
Let $H$ be a complete graph with $n$ vertices where every edge has weight $\frac1{n-1}$. We construct $G$ by attaching a path $P_v$ of length $\frac1\delta$ to each $v \in V(H)$, where the weight of each edge of each path is $1$.  Note that these paths are mutually disjoint. Since the induced subgraph $H$ of $G$ is an $1$-expander, we have $\lambda(\delta)\geq 1$.

To prove the proposition, it is sufficient to show that for any $\eps>0$,  if we run  \autoref{alg:findingexpander} on $G$, $\delta$ and $\eps$, then
all of the subsets of $V$ that we construct in the algorithm are $O(\delta)$-expanders. 
Let $S \subset V(G)$ be the set containing half of $V(H)$ together with the paths attached to its vertices. It is easy to see that $(S,\overline{S})$ is the minimum bisection (and the sparsest cut) of $G$. 
So even with an access to an oracle for the minimum bisection (or the sparsest cut) problem,
$V(G)$ will be divided into $S$ and $\overline{S}$ in the first step of Algorithm \ref{alg:findingexpander}. By a similar argument, it follows that in the second iteration, $V(G)$ will be divided into $4$ parts, where each of them contains a quarter of the vertices of $H$ together with the paths attached to them. Continuing this line of reasoning, at the end of the algorithm, $V(G)$ is divided into $2/\delta$ sets each with exactly $\delta/2$ fraction of the vertices of $H$ together with their attached paths. 
Depending on the value of $\eps$, the algorithm terminates at some iteration. But, since  all of the aforementioned sets are $O(\delta)$-expanders, the best set that the algorithm  finds is an $O(\delta)$-expander.
\end{proof}

\section{The SDP Rounding Algorithms}
\label{sec:sdprounding}
In this section, we prove our main results, theorems\ref{thm:maintheoremwithoutoutsideexpansion} and \ref{thm:mainwithoutsideexpansion}. 
Our proof follows the  plan that we discussed in \autoref{sec:proofoverview}.
Throughout this section, we assume $G$ is a $\Delta$-regular graph and $(\bx,\by,\lambda)$ represents a feasible solution of \ref{sdp:noh} or \ref{sdp:withh}. 
In the first step, we show that $G_x$ is a $\lambda/2$-small set weighted expander.
\begin{lemma}
\label{lem:expandingsmallsets}
For any $S \subset V$ of size at most $\frac{\delta n}{2}$, we have $\h_\bx(S) \geq \frac{\lambda}{2}$.
\end{lemma}
\begin{proof}
First we prove  $h_\by(S) \geq \frac{1}{2}$, and then by  constraint \eqref{eq:lambdaconstraint}, we conclude that $\h_\bx(S) \geq \frac{\lambda}{2}$.
We have 
\begin{eqnarray*}
h_\by(S)  =   \frac{\by(S,\overline{S})}{\bz(S)} &=&  \frac{1}{\bz(S)}\left(\sum_{ 
u \in S} \sum_{ v \in V} y_{\{u,v\}} - \sum_{u \in S} \sum_{v \in S} y_{\{u,v\}}\right) \\
&\geq &  \frac{1}{\bz(S)} \left(\sum_{u\in S}\sum_{ v  \in V} y_{\{u,v\}} - 
\frac{1}{\delta n} \sum_{u \in S}\sum_{v  \in S} \sum_{w \in V}  y_{\{u,w\}} \right) \\
&=& \frac{\delta n-|S|}{\delta n\cdot \bz(S)} \left(\sum_{u 
\in S}\sum_{w \in V} y_{ \{ u,w \} } \right) 
\end{eqnarray*}
where the first inequality uses  Constraint \eqref{eq:maxyconstraint}. 
Note that Constraint \eqref{eq:sumyconstraint} implies that $\sum_{u \in S} \sum_{w \in V} y_{\{u,w\}} \geq z(S)$. Combining it with the above inequalities and our assumption that $|S| \leq \delta\cdot \frac{n}{2}$, we get $h_\by(S) \geq \frac{1}{2}$. Therefore, to 
prove the lemma, it is enough to show that $\bx(S,\overline{S}) \geq 
\lambda\cdot \by(S,\overline{S})$. This directly follows from Constraint \eqref{eq:lambdaconstraint}. We have 
\begin{eqnarray} \bx(S,\overline{S}) =\bone_S^\intercal \left(\sum_{e} x_e L_{uv}\right)\bone_S   
\geq  \lambda \cdot \bone_S^\intercal \left(\sum_{u,v} y_{\{uv\}}L_{uv}\right)\bone_S = \lambda \cdot \by(S,\overline{S})
. 
\end{eqnarray}
So $\h_\bx(S) \geq \frac{\lambda}{2}$.
\end{proof}
In the next lemma, we provide an algorithm to find a set of vertices with small weighted expansion in $G_\bx$ and relatively small width. 
\begin{lemma}
\label{lem:nonexpandingset}
Let $\alpha := \frac{\sum_{e \in E} x_e}{\bz(V)}$. For any $w>0$, there is a set $S \subseteq V$ such that $$\h_\bx(S) \leq \dfrac{2\alpha(\log \, \frac{\Delta}{\alpha})}{w},$$ and $\frac{\max_{ u \in S} z_u}{\min_{u \in S} z_u} \leq e^{w}$. Furthermore, such a set can be found in polynomial time. 
\end{lemma}
\begin{proof}
Let $V_t =\Set{v \in V}{z_v = e^{t}}$ be the set of vertices with $x$-value $e^{t}$, and let  $V_{t_0,t_1}=
\Set{v \in V}{e^{t_0} \leq z_v \leq e^{t_1}}$. 
In addition, we define $V_{>t}$ and $V_{<t}$ to be $\Set{v}{ z_v > e^{t}}$ and $\Set{v}{z_v < e^{t}}$ respectively. 


It is sufficient to prove there is a $t \in \R $ such that 
\begin{equation}\label{eq:windowgoal}
h_\bx(V_{t,t+w})\leq \dfrac{2\delta(\log \, \frac{\Delta}{\delta})}{w}	
\end{equation}
 This proves the lemma since by definition of $V_{t,t+w}$ 
$$\frac{\max_{v \in V_{t,t+w} }z_v}{\min_{v \in V_{t,t+w}} z_v} \leq e^w.$$
In addition, since there are at most $n$ possible such sets, a simple linear time algorithm find the best $t$. Consider a  probability distribution with density function $p(t)\propto \bz(V_{t,t+w})$, for any $t \in \R$. To prove \eqref{eq:windowgoal}, it is enough to show  
\begin{equation}
\label{eq:expectationgoal}
\mathbb{E}_t [h_\bx(V_{t,t+w})]\leq \frac{2\delta(\log \, (\Delta/\delta))}{w}.
\end{equation}  
Intuitively, if $z_u$ and $z_v$ are close, then the  probability that 
$\{u,v\}$ is cut by a set $V_{t,t+w}$, which is essentially proportional 
to $ |\log z_u - \log z_v|$, is small. On the other hand, since 
$x_{\{u,v\}} \leq \min(z_u,z_v)$,  as $z_u$ and $z_v$ gets further, the 
relative contribution of $x_{\{u,v\}}$, $\frac{x_{\{u,v\}}}{\max (z_u,z_v)}$, 
decreases. 
We start by upper bounding $\EE{t}{\h_\bx(V_{t,t+w})}$.
\begin{eqnarray*}
\mathbb{E}_{t}[h_\bx(V_{t,t+w})] 
=&  \displaystyle\int_{-\infty}^{\infty} p(t)\dfrac{\bx(V_{t,t+w},V \setminus {V_{t,t+w}})}{\bz(V_{t,t+w})}\,dt \\ 
=& \dfrac{1}{Z} \displaystyle\int_{-\infty}^{\infty} \bx(V_{t,t+w},V \setminus {V_{t,t+w}})\,dt\\
\leq & \dfrac{2}{Z} \sum_{u,v}|\log z_u - \log z_v|x_{\{u,v\}}
\end{eqnarray*}
where $Z=\int_{-\infty}^{\infty} \bz(V_{t,t+w})dt$ is the normalizing 
constant of the probability distribution. The last inequality holds, 
since an edge $\{u,v\}$ appears in $E(V_{t,t+w},V\setminus{V_{t,t+w}})$ 
only when exactly one of  the numbers $\log z_u$ and $\log z_v$ lies 
in the interval $[t,t+w]$.
It is fairly easy to verify  $Z = w \cdot \bz(V)$. Substituting $Z$ 
into above, to show \eqref{eq:expectationgoal}, it is enough to prove 
that 
\begin{equation}\label{eq:mainineq} \sum_{u,v}|\log \frac{ z_u } 
{z_v}|x_{\{u,v\}} \leq \delta(\log \, \frac{\Delta}{\delta}) \bz(V).
\end{equation}

To prove \eqref{eq:mainineq}, it is enough to show an analogous 
statement for  every vertex $u \in V$. Assume there is an ordering on the vertices of the graph such that $u < v$ implies $z_u \leq z_v$ and set $\alpha_u :=\sum_{v: 
v < u} x_{\{u,v\}}/z_u$. For any vertex $u \in V$, we show  
\begin{equation}
\label{eq:singlevertexcont}
\sum_{v: v  < u}\left|   \log \frac {z_u}{ z_v}\right|
x_{\{u,v\}}  \leq \alpha_u(\log \, \frac{\Delta}{\alpha_u}) z_u. 
\end{equation}
First, we show that by summing up \eqref{eq:singlevertexcont} over all vertices, we obtain \eqref{eq:mainineq}. Then we prove \eqref{eq:singlevertexcont}. Observe that summing up LHS of \eqref{eq:singlevertexcont} over all $u \in V$, gives the LHS of \eqref{eq:mainineq}. Therefore, it is sufficient to show 
\begin{equation}
\label{eq:singlevertexcontribution2}
\sum_{u \in V} \alpha_u(\log \, \frac{\Delta}{\alpha_u}) z_u \leq 
\alpha(\log \, \frac{\Delta}{\alpha}) \bz(V)
\end{equation}
We prove this by Jensen's inequality.
Since $f(s)= s (\log \frac{\Delta}{s})$ is a concave function, by Jensen's inequality we have  \\
\begin{eqnarray*}
\sum_{u} \frac{z_u}{\bz(V)}\cdot\left(\alpha_u \log \frac{\Delta}{\alpha_u} \right) &\leq & \frac{\sum_u z_u \alpha_u}{\bz(V)}\log \frac{\Delta}{\frac{\sum_u z_u \alpha_u}{\bz(V)}}\\  &=& 
\frac{\sum_{ u \in V} \sum_{v: v< u} x_{\{u,v\}}}{\bz(V)} \log \frac{\Delta \cdot \bz(V)}{\sum_{ u \in V} \sum_{v:  v< u} x_{\{u,v\}}} =   \alpha \log \frac{\Delta}{\alpha},
\end{eqnarray*}
where the first and second equality use definitions of $\alpha_u$ and $\alpha$, respectively. This proves \eqref{eq:singlevertexcontribution2} which implies that by summing up \eqref{eq:singlevertexcont} over all vertices we get \eqref{eq:mainineq}. It remains to prove \eqref{eq:singlevertexcont}. By definition of $z_v$, to prove  \eqref{eq:singlevertexcont}, we can show \\
\begin{equation*}
\sum_{v: v  < u}\left|   \log \frac {z_u}{ x_{\{u,v\}}}\right|x_{\{u,v\}} \leq \alpha_u(\log \, \frac{\Delta}{\alpha_u}) z_u
\end{equation*}
By definition of $\alpha_u$,$\{\frac{x_{\{u,v\}}}{\alpha_u z_u}\}_{v<u}$ is a probability distribution on neighbors $v$ of $u$ where $v<u$, so we can rewrite the LHS in terms of the entropy of this distribution, as follows: 
\begin{eqnarray*}
\sum\limits_{v<u} \left(\log \frac{z_u}{x_{\{u,v\}}}\right)x_{\{u,v\}} &= & 
(\log \frac{1}{\alpha_u})\sum\limits_{v<u} x_{\{u,v\}}+ \alpha_u z_u \left(\sum_{v < u}\left(\log \frac{\alpha_u z_u}{ x_{\{u,v\}}}\right) \frac{x_{\{u,v\}}}{\alpha_u z_u}  \right) \\ &\leq &   \alpha_u z_u(\log \frac{1}{\alpha_u}) + \alpha_u z_u \log \Delta = \alpha_u z_u\log \frac{\Delta}{\alpha_u}
\end{eqnarray*}
where the inequality holds since $u$ has at most $\Delta$ neighbors and consequently  the entropy of the distribution defined above is at most $\log \Delta$.  
As stated before it  proves \eqref{eq:singlevertexcont} and 
finishes the proof of the lemma.
\end{proof}


\begin{lemma}
\label{lem:findingexpander}
Given $S \subseteq V(G)$, for any $0 < \epsilon < 1$, there is a set $T \subseteq S$ satisfying one of the following cases. 
\begin{enumerate}[i)]
\item \label{case:expander} $|T| \geq \frac{3\delta n}{8}$ and $\lambda_2(G[T]) \geq \epsilon \cdot \Delta$.
\item \label{case:nonexpanding} $|T| \leq \frac{\delta n}{2}$ and  \begin{equation}
\label{eq:nonexpandingset}
\h_\bx(T) \leq 2w\cdot\log \frac{1}{\delta}\sqrt{2\epsilon}\cdot \Delta+ \h_\bx(S),\end{equation} 
where $w = \frac{\max_{v \in S} z_u}{\min_{v \in S} z_v}$. 
\end{enumerate}
\end{lemma}

\begin{proof}
If $|S| \leq \delta \cdot n /2$, $S$ satisfies Case \ref{case:nonexpanding} and we are done. Otherwise, we set $\delta' = \frac{\delta n}{|S|}$ and run  \autoref{alg:findingexpander} on input subgraph $G[S]$, $\delta'$ and $\epsilon$. If it returns a set $T \subseteq V(G)$, then we are in Case \ref{case:largeexpander} of  \autoref{lem:graphdecomposition} which implies we have found the desired expander. 
Now, assume the output of the algorithm is Case 
\ref{case:decomposition}, a partition $\mathcal{P}$ of $S$ satisfying \eqref{eq:numberofedges}. Since by 
 \autoref{alg:findingexpander}, any element of $\mathcal{P}$ has at 
most $\delta' \cdot |S| /2 = \delta \cdot n /2$ vertices, to prove the lemma, it suffices to 
show that there exists a set $T \in 
\mathcal{P}$ for which \eqref{eq:nonexpandingset} holds. To show it, we consider a probability distribution on elements of $\mathcal{P}$ where for every $T \in \mathcal{P}$, $\mathbb{P}(T) \propto \bz(T)$ and prove that 
\begin{equation}
\label{eq:expectedexpansion}
\mathbb{E}_{T}[h_\bx(T)] \leq 2w(\log \frac{1}{\delta})\sqrt{2\epsilon}\cdot \Delta + \h_\bx(S).
\end{equation} 
We can write $\mathbb{E}_T [\h_\bx(T)]$ as follows: 
\begin{eqnarray*}
\mathbb{E}_T [\h_\bx(T)] &= & \sum_{T \in \mathcal{P}} \mathbb{P}[T]\frac{\bx(T,\overline{T})}{\bz(T)} 
 = \frac{1}{\bz(S)} \sum_{T \in \mathcal{P}} \bx(T,\overline{T}) \\
&=& \frac{1}{\bz(S)} \left(\sum_{T \in \mathcal{P}} \bx(T,S\setminus{T})+ \sum_{T \in \mathcal{P}} \bx(T,\overline{S})\right) \\
&= &\frac{1}{\bz(S)} \sum_{T \in \mathcal{P}} \bx(T,S\setminus{T}) + \h_\bx(S)    
\end{eqnarray*}
So comparing to our goal, \eqref{eq:expectedexpansion}, we only need to prove 
\begin{equation}
\label{eq:sumofweightsofedges}
\sum_{T \in \mathcal{P}} \bx(T,S\setminus{T}) \leq 2w(\log \frac{1}{\delta})\sqrt{2\epsilon}\cdot \Delta \cdot \bz(S).
\end{equation}
Note that it simply follows from \eqref{eq:numberofedges} and 
\begin{equation*}
\frac{\sum_{T \in \mathcal{P}} \bx(T,S\setminus{T})}{\bz(S)}\leq w
\frac{\sum_{T \in \mathcal{P}} |E(T,S\setminus{T})|}{|S|},
\end{equation*}
which is implied by definition of $w$.
\end{proof} 
It is easy to see that Theorems \ref{thm:mainwithoutsideexpansion} and \ref{thm:maintheoremwithoutoutsideexpansion} follow from the above three lemmas. 

\mainthm*
\begin{proof}
We combine Lemmas \ref{lem:nonexpandingset}, \ref{lem:findingexpander}, and \ref{lem:expandingsmallsets} to prove the theorem. Let $w,\epsilon >0$ be two parameters that we will fix later. 
First, by  \autoref{lem:nonexpandingset}, we find a set 
$S$ with width $e^w$ such that 
\begin{equation}\label{eq:hxSissmall}h_\bx(S) \leq \frac{2\alpha(\log \, 
\frac{\Delta}{\alpha})}{w}\leq 2\frac{\Delta}{e\cdot w}
\end{equation}
where the last 
inequality holds since $\alpha \log \frac{\Delta}{\alpha}\leq 
\frac{\Delta}{e}$ for any $\alpha>0$. Then, we run the algorithm in  
\autoref{lem:findingexpander} 
on subgraph $G[S]$ and parameters $\epsilon$ and $\delta$. Let $T 
\subseteq S$ be the  output. We choose $\epsilon$ and 
$w$ such that $2e^w|\log\delta|\sqrt{2\epsilon}\cdot \Delta
+2\frac{\Delta}{ew} < \lambda/2$. 
This implies  Case 
\ref{case:expander} of  \autoref{lem:findingexpander} is satisfied; this is because Case \ref{case:nonexpanding} implies 
\begin{eqnarray*}h_\bx(T) &\leq& 2e^w|\log\delta|\sqrt{2\eps}\cdot\Delta+h_\bx(S)\\
	&\leq&  2e^w|\log\delta|\sqrt{2\epsilon}\cdot \Delta
+2\frac{\Delta}{ew} < \lambda/2,
\end{eqnarray*}
 which contradicts  \autoref{lem:expandingsmallsets} as $|T|\leq\delta n/2$. The second inequality in the above follows by \eqref{eq:hxSissmall}. Letting
$$\epsilon =\frac{{\lambda}^2}{129\Delta^2(\log^2 \delta)\cdot \exp(\frac{16\Delta}
{e\cdot\lambda})}, \, \, w= \frac{8\Delta}{e\cdot \lambda},$$
we get $2e^w|\log\delta|\sqrt{2\epsilon}\cdot \Delta
+2\frac{\Delta}{ew} < \lambda/2$.
Therefore, by Case \ref{case:expander} of \autoref{lem:findingexpander}, $|T| \geq 3\delta \cdot n/8$ and 
$$\lambda_2(L_{G[T]}) \geq \epsilon\Delta\gtrsim \frac{\lambda}{\log^2\delta\cdot\exp(O(\Delta/\lambda))}$$
as desired. In the second equation we absorbed the term $\Delta/\lambda$ in the denominator in $\exp(O(\Delta/\lambda))$.
\end{proof}

Using similar ideas combined with the constraint \eqref{eq:hconstraint} of \ref{sdp:withh}, we can prove  \autoref{thm:mainwithoutsideexpansion}.

\thmmainnoexp*
\begin{proof}
The structure of the proof is very similar to the proof of  
\autoref{thm:maintheoremwithoutoutsideexpansion}. Again, we use  
\autoref{lem:nonexpandingset} to find a set $S \subseteq V$ with 
width $e^w$, and run the algorithm in \autoref{lem:findingexpander} on  $G[S]$, $\delta$ and a proper value of 
$\epsilon$. The main 
difference  is to use Constraint \eqref{eq:hconstraint} of 
the \ref{sdp:withh} to prove a stronger upper bound on the weighted expansion of $S$, 
\begin{equation}\label{eq:betterexpansionS}h_\bx(S)\leq \frac{(\Delta-h^*)}{w}\log \frac{2\Delta}{(\Delta-h^*)}	
\end{equation}
First, recall that by \autoref{lem:nonexpandingset}, $h_\bx(S) 
\leq \frac{2\alpha(\log \, \frac{\Delta}{\alpha})}{w}$
where $\alpha =\frac{\sum_{e \in E} x_e}{\bz(V)}$.
It follows by Constraint \eqref{eq:hconstraint} (and  $z_v=\max_{e\sim v} x_e$) that 
$$\alpha=\frac{\sum_e x_e}{\sum_v \max_{e\sim v} x_e}\geq \frac{(\Delta-h^*)}{2}.$$
To prove \eqref{eq:betterexpansionS}, it is enough to note that $\alpha\log\frac{\Delta}{\alpha}$ is a decreasing function of $\alpha$ for $\alpha\geq \Delta/e$, and $(\Delta-h^*)/2\geq \Delta/e$ as $h^*\leq \Delta(1-2/e)$ by the lemma's assumption. Therefore, 
$$ h_\bx(S) \leq \frac2w\alpha\log\frac{\Delta}{\alpha} \leq \frac{(\Delta-h^*)}{w}\log\frac{2\Delta}{\Delta-h^*}.$$
Similar to \autoref{thm:maintheoremwithoutoutsideexpansion}, if we choose $\eps,w$ such that 
\begin{equation}\label{eq:hxSbetterlambda}|\log\delta|\sqrt{2\eps}\cdot\Delta+h_\bx(S) 
<\lambda/2,	
\end{equation}
then (by an application of \autoref{lem:expandingsmallsets}) Case \ref{case:expander} of \autoref{lem:findingexpander} is satisfied.
Letting $$w= \frac{ 4(\Delta-h^*)(\log \frac{2\Delta}{\Delta-
h^*})}{\lambda}, \hspace{3ex} \epsilon =\frac{{\lambda}^2}
{129\Delta^2(\log^2 \delta)\cdot (\Delta/(\Delta-
h^*))^{\frac{32(\Delta-h^*)}
{\lambda}}},$$  
and using \eqref{eq:betterexpansionS}, it is easy to see that \eqref{eq:hxSbetterlambda} is satisfied. 
Therefore, by Case \ref{case:expander} of \autoref{lem:findingexpander}, $|T|\geq 3\delta \cdot n /8$ and 
$$\lambda_2(L_{G[T]}) \geq \eps\cdot\Delta\gtrsim \frac{\lambda^2}{\Delta\log^2\delta(1+2h^*/\Delta)^{O(\Delta/\lambda)}}$$
as desired. In the second inequality we use that $\frac{\Delta}{\Delta-h^*}\leq 1+2h^*/\Delta$ for $\h^*\leq\Delta(1-2/e)$.
\end{proof}

\section{Integrality Gap}
In this section we prove that the integrality gap of \ref{sdp:noh} is  $\Omega(\min((\log \frac{1}{\delta}),\Delta))$. 
\begin{theorem}
\label{thm:IGfirst}
For any integer $\Delta > 0$ and $\delta=2^{-\Delta}$, there exists an $\Delta$-regular graph $G$ such that  $\lambda(\delta) \leq O(1)$, but the optimal value of \ref{sdp:noh} is at least $\Omega(\Delta)$.
\end{theorem}

\begin{proof}
Let $H$ be a $\Delta$-dimensional hypercube with $2^\Delta$ vertices. We let $n$ be a  sufficiently large multiple of $2^\Delta$ and construct $G$ as follows:   
We blow up every vertex $i \in V(H)$ by a cloud of  $\frac{ 
n}{2^\Delta}$ vertices, called $B_i$. For every edge $\{i,j\} \in E(H)$, 
we place a complete  bipartite graph between $B_i$ and $B_j$, where 
the weight of every edge is $\frac{1}{\delta n}$\footnote{Here for the sake of 
simplicity, we construct a weighted graph $G$, but one can extend the construction to unweighted graphs by replacing the weighted complete bipartite graphs with constant degree expanders.}.
By definition, $G$ is a $\Delta$-regular graph. First, we show $\lambda(\delta)=O(1)$ and then we 
build a feasible solution of \ref{sdp:noh} of value $
\Omega(\Delta)$.

For every $S \subseteq V(G)$, we prove $\h(G[S]) = O(1)$, which 
by Cheeger's inequality (\autoref{thm:cheeger}) implies $\lambda_2(L_{G[S]}) = O(1)$ and 
consequently $\lambda(\delta) = O(1)$.
 Without loss of generality, 
assume there is a  dimension cut $(T,\overline{T})$ of $H$ such that the union of clouds of vertices of $T$ cut $S$. Let $B_T=\cup_{i\in T} B_i$. Since for each vertex $u\in B_T$, only $1/\Delta$ fraction of edges incident to $u$ are leaving $B_T$, we have 
$$\h_{G[S]}(B_T\cap S)\leq \h(B_T\cap S) = O(1).$$
Similarly, $\h_{G[S]}(B_{\overline{T}}\cap S) = O(1)$; so $\h(G[S])=O(1)$. 

It remains to  present a feasible 
solution for \ref{sdp:noh}  of value $\Omega(\Delta)$. We  
construct $\bx,\by$ as follows: \\ 
\begin{equation*}
\begin{aligned}
 x_{\{u,v\}}&=& &1&  \hspace{3ex} &\forall u \sim v&   \\
 y_{\{u,v\}} &=& &\frac{1}{\delta n}& \hspace{3ex} &\forall i \in V(H), \, \, \forall u,v \in B_i& \\
\end{aligned}
\end{equation*}
With this solution, the only non-trivial constraint of  \ref{sdp:noh} that we should verify is the first constraint, i.e,
\begin{eqnarray}
\label{eq:SDPmainconstraint}
L_{G_{\bx}} \succeq \Omega(\Delta)L_{G_{\by}}
\end{eqnarray}
Note that since $x_{\{u,v\}}=1$ for all $\{u,v\}\in E(G)$, $G_\bx=G$ (and $L_{G_\bx}=L_G$).
Let $\{i,j\}$ be an edge of $H$. Since $G[B_i \cup B_j]$ is a complete bipartite graph, we have 
\begin{eqnarray}
\label{eq:hypercubeexpansion}
L_{G_\bx[B_i \cup B_j]} \succeq \Omega\left(\frac{L_{K_{2\delta n}}}{\delta n}\right)   \succeq \Omega(L_{G_\by[B_i \cup B_j]})
\end{eqnarray}

Rewriting the above inequality by extending $L_{G_\bx[B_i \cup B_j]}$ and $L_{G_\by[B_i \cup B_j]}$ to $\tilde{L}_{G_\bx[B_i\cup B_j]},\tilde{L}_{G_\by[B_i\cup B_j]}\in \R^{V\times V}$, by inserting zero rows and columns corresponding to vertices in $\overline{B_i\cup B_j}$, we get
$$ \tilde{L}_{G_\bx[B_i \cup B_j]} \succeq \Omega(\tilde{L}_{G_\by[B_i\cup B_j]}).$$ 
Summing up the above inequality  over all $\{i,j\} \in E(H)$ gives 
\eqref{eq:SDPmainconstraint}.
\end{proof}

\section{Discussion}
We provide the first approximation algorithms for the largest induced expander problem. Let us conclude by providing several open problems and future directions. Firstly, we can show that the  exponential dependency on $\Delta/\lambda$ in \autoref{thm:maintheoremwithoutoutsideexpansion} is necessary to our rounding algorithm. But, we are not aware of any tight integrality gap example. It is a fascinating question if this dependency can be improved to $\poly(\Delta/\lambda)$. Secondly, our techniques fail to find induced expanders in dense regular graphs when $\Delta$ is significantly larger than $\delta n$; in such cases, one can construct a trivial integral SDP solution for any given graph $G$. A resolution of this question can lead to new approximation algorithms for the hidden clique problem. Perhaps a practical downside of our algorithm is the need to solve a semidefinite program. It is interesting if one can reproduce our results using fast spectral methods. 
\bibliographystyle{alpha}
\bibliography{references}

%


\end{document}